\documentclass[a4paper,USenglish,cleveref, autoref, thm-restate]{lipics-v2021}

\input xy
\xyoption{all}
\UseComputerModernTips
\nolinenumbers

\newcommand{\set}{\ensuremath{\mathsf{Set}}}
\newcommand{\cat}{\ensuremath{\mathsf{Cat}}}
\newcommand{\grp}{\ensuremath{\mathsf{Grp}}}
\newcommand{\grpd}{\ensuremath{\mathsf{Grpd}}}
\newcommand{\xmod}{\ensuremath{\mathsf{XMod}}}
\newtheorem*{proposition*}{Proposition}
\newtheorem*{lemma*}{Lemma}

\newcommand{\Z}{\mathcal{Z}}

\newcommand{\TJ}{\T^\J}
\newcommand{\aut}{\mathsf{Aut}}
\newcommand{\defaut}{\mathsf{DefInn}}
\newcommand{\B}[1]{\mathcal{B}(#1)}
\newcommand{\C}{\mathbb{C}}
\newcommand{\D}{\mathbb{D}}
\newcommand{\E}{\mathcal{E}}
\newcommand{\F}{\mathcal{F}}

\newcommand{\J}{\mathcal{J}}
\newcommand{\T}{\mathbb{T}}
\newcommand{\mS}{\mathbb{S}}
\newcommand{\la}{\langle}
\newcommand{\ra}{\rangle}
\newcommand{\x}{\mathsf{x}}
\newcommand{\y}{\mathsf{y}}
\newcommand{\z}{\mathsf{z}}
\newcommand{\Sort}{\mathsf{Sort}}

\newcommand{\Term}{\mathsf{Term}}
\newcommand{\Horn}{\mathsf{Horn}}
\newcommand{\Tmod}{\T\mathsf{mod}}
\newcommand{\Smod}{\mS\mathsf{mod}}
\newcommand{\inv}{\mathsf{Inv}}
\newcommand{\perm}{\mathsf{Perm}}
\newcommand{\id}{\mathsf{Id}}
\newcommand{\dom}{\mathsf{dom}}
\newcommand{\cod}{\mathsf{cod}}
\newcommand{\StrMonCat}{\mathsf{StrMonCat}}
\newcommand{\mon}{\mathsf{Mon}}
\newcommand{\ob}{\mathrm{Ob}}
\newcommand{\arr}{\mathrm{Arr}}
\newcommand{\fpTmod}{\mathsf{fp}\Tmod}
\newcommand{\tensor}{\otimes}

\newcommand{\TJmod}{\TJ\mathsf{mod}}
\newcommand{\Aut}{\mathsf{Aut}}

\title{Polymorphic Automorphisms and the Picard Group} 

\titlerunning{Polymorphic Automorphisms and the Picard Group} 

\author{Pieter Hofstra}{Dept. of Mathematics \& Statistics,
University of Ottawa, Ottawa, Ontario, Canada \and \emph{phofstra@uottawa.ca} }{}{}{Research funded by an NSERC Discovery Grant}

\author{Jason Parker}{Department of Mathematics \& Computer Science,
Brandon University,  
Brandon, Manitoba, Canada \and \emph{parkerj@brandonu.ca} }{}{}{Postdoctoral research funded by NSERC grant of  R. Lucyshyn-Wright (Brandon)}

\author{Philip J. Scott\footnote{corresponding author}}{Dept. of Mathematics \& Statistics, University of Ottawa, Ottawa, Ontario, Canada   
\and \emph{philip.scott@uottawa.ca}
}{}{}{Research funded by an NSERC Discovery Grant}

\authorrunning{P.\,Hofstra, J.\, Parker,  and P. J.\, Scott} 

\Copyright{Pieter Hofstra, Jason Parker, and Philip J. Scott} 



\category{} 

\relatedversion{} 

\acknowledgements{Pieter Hofstra would like to acknowledge illuminating discussions with Martti Karvonen and Eugenia Cheng.}

\EventEditors{John Q. Open and Joan R. Access}
\EventNoEds{2}
\EventLongTitle{42nd Conference on Very Important Topics (CVIT 2016)}
\EventShortTitle{CVIT 2016}
\EventAcronym{CVIT}
\EventYear{2016}
\EventDate{December 24--27, 2016}
\EventLocation{Little Whinging, United Kingdom}
\EventLogo{}
\SeriesVolume{42}
\ArticleNo{23}

\begin{document}

\maketitle

\begin{abstract}
We investigate the concept of definable, or inner, automorphism in
the logical setting of partial Horn theories. The central technical result extends
a syntactical characterization of the group of such automorphisms (called the covariant isotropy group)
associated with an algebraic theory to the wider class of quasi-equational theories. 
We apply this characterization to prove that the isotropy group of a strict monoidal category
is precisely its Picard group of invertible objects. Furthermore, we obtain 
an explicit description of the covariant isotropy group of a presheaf category.
\end{abstract}

\section{Introduction}
\label{sec:intro}

In algebra, model theory, and computer science, one encounters the notion of \emph{definable automorphism} 
(the nomenclature varies by discipline).
In first-order logic for example (see e.g.~\cite{Shoenfield}), 
an automorphism $\alpha$\/ of a model $M$\/ is called \emph{definable} (with parameters in $M$) when there is a formula $\varphi(x,y)$\/
in the ambient language (possibly containing constants from $M$) such that for all $a,b \in M$\/ we have
\[ \alpha(a)=b \Longleftrightarrow M \models \phi(a,b). \]
The case of groups is instructive: for a group $M$, consider the 
formula $\phi(x,y)$\/ given as
\[ \phi(x,y): \; y=c^{-1}xc \]
for some $c \in M$. This defines an (inner) automorphism of $M$. 
Note that in this case the automorphism is also determined by a term $t(x):=c^{-1}xc$\/ via $a \mapsto t(a)$.

 These definable automorphisms
have various interesting aspects: first of all, they are in some sense \emph{polymorphic} or uniform. This means roughly that
the same term $t$, possibly after replacing constants from $M$, can also define an automorphism of
another model $N$. Secondly, the definable automorphisms can also provide a generalized notion of \emph{inner automorphism}, even for
theories where it does not make sense to speak of group-theoretic conjugation. Indeed, Bergman~\cite[Theorem 1]{Bergman} shows that in the category of groups,
the definable group automorphisms, i.e. the inner automorphisms given by conjugation, 
can be characterized purely \emph{categorically} by the fact that they extend naturally along any homomorphism. That is: an automorphism
$\alpha:G \xrightarrow{\sim} G$\/ is inner precisely when for any homomorphism $m:G \to H$\/ there is an extension $\alpha_m : H \xrightarrow{\sim} H$\/ making diagram (a) commute
and also making
\[
(a) \  \xymatrix{
G \ar[d]_\alpha \ar[r]^m & H \ar[d]^{\alpha_m} \\
G \ar[r]_m & H
} \qquad\qquad
(b) \ \xymatrix{
H \ar[d]_{\alpha_m} \ar[r]^n & K \ar[d]^{\alpha_{nm}} \\
H \ar[r]_n & K
}
\]
diagram (b) commute for any further homomorphism 
$n : H \to K$, so that in particular $\alpha = \alpha_{\mathsf{id}_G}$ by diagram (a). 
If $\alpha$\/ is conjugation by $g \in G$, then $\alpha_m$\/ is conjugation by $m(g) \in H$. 
Conversely, given any system of group automorphisms $\{\alpha_m : H \xrightarrow{\sim} H \mid m:G \to H\}$ 
with $\alpha = \alpha_{\mathsf{id}_G}$ that makes diagrams (a) and (b) commute, 
Bergman shows that there is a unique element $s \in G$ such that $\alpha$ is given by conjugation with $s$. 
Bergman therefore refers to such a system $\{\alpha_m \mid m:G \to H\}$\/ as an \emph{extended inner
automorphism} of $G$. 

In categorical logic, we have a canonical method for studying this phenomenon. To any category $\C$, we may associate the functor
\begin{equation}\label{eq:isotropyfunctor}
\Z_\C:\C \to \grp \; ; \qquad \Z_\C(C) := \aut(\pi:C/\C \to \C).
\end{equation}
Let us unpack this. We have the co-slice category $C/\C$\/ whose objects are maps $C \to D$ and whose arrows are commutative triangles. 
The projection functor $\pi:C/\C \to C$\/ sends $C \to D$\/ to $D$. We then consider the group of natural automorphisms of this projection functor, i.e. 
the group of \emph{invertible} natural transformations $\alpha: \pi \Rightarrow \pi$. To give such an $\alpha$\/ is equivalent to giving, for each object $m:C \to D$\/ of $C/\C$,
an automorphism $\alpha_m:D \xrightarrow{\sim} D$, subject to the naturality condition that for any composable pair $m : C \to D, n : D \to E$ in $\C$, we have $\alpha_{nm}n=n\alpha_{m}$ as in diagram (b) above. 
Thus, in Bergman's terminology, $\Z_\C(C)$\/ is the group of extended inner automorphisms of $C$. We call $\Z_\C$\/ the \emph{(covariant) isotropy group (functor)} of $\C$.
Another useful way of thinking about this group is by noticing that the assignment $C \mapsto \aut(C)$\/ is generally not functorial, unless $\C$ is a groupoid. The isotropy group
offers a remedy: the assignment $C \mapsto \Z_\C(C)$\/ \emph{is} functorial, as is straightforward to check, and for each $C$\/ there is a comparison homomorphism 
\begin{equation}\label{eq:comp}
\theta_C: \Z_\C(C) \to \aut(C) \; ; \qquad \alpha \mapsto \alpha_{\mathsf{id}_C}
\end{equation}
that sends an extended 
inner automorphism $\alpha$\/ to its component at the identity of $C$.\footnote{
P. Freyd~\cite{Freyd} studied a somewhat similar notion while modelling Reynolds' parametricity for parametric polymorphism.
 As a special case, his work leads to a \emph{monoid} of natural endomorphisms
 of the projection functor, whereas in our case, we would obtain the
 subgroup of invertible elements in this monoid.} We can then turn Bergman's aforementioned
  result for the category $\grp$ into a \emph{definition} for an arbitrary category $\C$, by defining an 
  automorphism $f : C \xrightarrow{\sim} C$ of an object $C \in \C$ to be \emph{inner} just if $f$ is in 
  the image of $\theta_C : \Z_\C(C) \to \aut(C)$. Less precisely, the automorphism $f : C \xrightarrow{\sim} C$ is inner if it can be coherently extended along any arrow out of $C$. 

(For readers familiar with topos theory and/or earlier papers on the subject of isotropy groups, we point out that in~\cite{FHS, FHK} we consider instead the \emph{contravariant}
isotropy groups $\aut(\pi : \C/C \to \C)$. Now if $\mathbb{T}$\/ is a suitable logical theory with classifying topos $\B{\mathbb{T}}$, then (a restriction of) 
the contravariant isotropy group of $\B{\mathbb{T}}$
coincides with the covariant isotropy group of the category $\fpTmod$\/ of finitely presented $\mathbb{T}$-models. Moreover, 
calculation of the latter group generally also yields a description of the covariant isotropy group of the larger category $\Tmod$\/ of \emph{all} $\mathbb{T}$-models,
which is our focus in the present paper.)

In~\cite{HPS}, the case where $\C$\/ is the category of models of an equational theory is analysed. Among other things, a complete syntactic characterization 
of covariant isotropy for such a $\C$ is obtained, recovering not only Bergman's result for 
$\C = \grp$ but also characterizing the definable automorphisms of other common algebraic structures such 
as monoids and rings. 
 In applying the general characterization in specific instances, one typically needs to analyse the result of adjoining one or more indeterminates
to a given model, and this in turn leads one to consider the \emph{word problem} for such models. 

The present paper, which is based on the PhD research~\cite{thesis} of the second author,
 is concerned with the analysis of the notion of isotropy or definable automorphism for (strict) monoidal categories and related structures. 
It hardly needs arguing that monoidal categories play various important roles in mathematics and theoretical computer science, both as objects of study in their own right,
as models of logical theories, and as basic tools for studying other phenomena. 
However, we should point out here an observation by Richard Garner~\cite[Proposition 3]{Garner} to the effect that both $\cat$\/ and $\grpd$,
the categories of small categories and small groupoids respectively, have \emph{trivial} covariant isotropy, in the sense that for any category/groupoid $\C$ we have 
$\Z(\C)=1$, the trivial group. As such, it is perhaps surprising that the category of strict monoidal categories
has \emph{non-trivial} isotropy. In fact, and this is the central result of the present paper, 
the isotropy group of a strict monoidal category is precisely its \emph{Picard group} (its
group of invertible objects). 

Since the theory of strict monoidal categories is not a purely equational theory, we cannot directly use results from~\cite{HPS}. Instead, we need to work in the setting of
 \emph{quasi-equational theories}. These are multi-sorted theories in which the operations can be \emph{partial}; equivalently, they are
finite-limit theories. These include the theories of categories, groupoids, strict monoidal categories, symmetric/braided/balanced monoidal categories, and crossed modules.
They also include what one might call \emph{diagram theories}, which are theories describing diagrams of a specified shape in a category of models. As a special case, one obtains theories whose categories of models are presheaf categories.\footnote{Not to be confused with the so-called \emph{theories of presheaf type}, which
are theories whose classifying topos happens to be a presheaf topos.} Our first main contribution of the paper is then a generalization of the syntactic characterization of
isotropy from equational theories to this wider class of quasi-equational theories.

While we indicated why the non-trivial isotropy of strict monoidal categories is perhaps surprising, 
there is also a sense in which it is to be expected. Indeed, since strict monoidal categories
are monoids internal to $\cat$, we expect that the isotropy of strict monoidal categories is closely related to that of monoids. 
Since the isotropy of a monoid $M$\/ is its subgroup of 
invertible elements, the conjecture that the isotropy of a strict monoidal category is its group of invertible objects is not unreasonable. 
However, it is not at all immediate that the isotropy of
a strict monoidal category should be determined \emph{completely} by its set of objects; 
the recognition that this \emph{is} the case is the second main contribution of this paper. 

\emph{A priori}, one can try to establish this result in a variety of ways. 
First of all, it can be approached purely syntactically, making careful analysis of
the word problem for strict monoidal categories. 
However, several aspects of this analysis can also be cast in more conceptual terms, giving 
rise to a categorical way of deriving 
the isotropy of strict monoidal categories from that of monoids. 
We thus also include a more categorical viewpoint,  
which  applies to several other theories 
of categorical structures, including crossed modules. 

\section{Quasi-equational theories}\label{sec:QE}
We begin by reviewing the relevant notions from categorical logic. For more 
details concerning quasi-equational theories and partial Horn logic, we refer to~\cite{Palmgren}.
 For a general treatment of categorical logic, see~\cite{Pitts}. 

\begin{definition}[\textbf{Signatures, Terms, Horn Formulas, Horn Sequents, Quasi-Equational Theories}]
{\em 
\
\begin{itemize}
\item A \emph{signature} $\Sigma$ is a pair of sets $\Sigma = (\Sigma_{\mathsf{Sort}}, \Sigma_{\mathsf{Fun}})$, 
where $\Sigma_{\mathsf{Sort}}$ is the set of \emph{sorts} of $\Sigma$ and $\Sigma_{\mathsf{Fun}}$ is the set 
of \emph{function/operation symbols} of $\Sigma$. Each element $f \in \Sigma_{\mathsf{Fun}}$ comes equipped with a 
finite tuple of sorts $(A_1, \ldots, A_n, A)$, and we write $f : A_1 \times \ldots \times A_n \to A$.

\item Given a signature $\Sigma$, we assume that we have a countably infinite set of variables 
of each sort $A$. Then one can recursively define the set $\Term(\Sigma)$ of \emph{terms} 
of $\Sigma$ in the usual way, so that each term will have a uniquely defined sort. We write 
$\Term^c(\Sigma)$ for the set of \emph{closed} terms of $\Sigma$, i.e. terms containing no variables. 

\item Given a signature $\Sigma$, one can recursively define the set $\Horn(\Sigma)$ of \emph{Horn formulas} of $\Sigma$\/
 in the usual way, where a Horn formula is a finite conjunction of equations between elements of $\Term(\Sigma)$. We write $\top$ for the empty conjunction. 

\item A \emph{Horn sequent} over a signature $\Sigma$ is an expression of the form 
$\varphi \vdash^{\vec{x}} \psi$, where $\varphi, \psi \in \mathsf{Horn}(\Sigma)$ and have variables among $\vec{x}$. 

\item A \emph{quasi-equational theory} $\T$ over a signature $\Sigma$ is a set of Horn sequents over $\Sigma$, which we call the \emph{axioms} of $\T$. \qed
\end{itemize}
}
\end{definition}

\noindent One can set up a deduction system of \emph{partial Horn logic} (PHL) for quasi-equational theories, axiomatizing the notion of a \emph{provable sequent}
$\varphi \vdash^{\vec{x}} \psi$\/. Accordingly, for a theory $\T$\/ we have the notion of a 
$\T$-provable sequent; moreover, if $\top \vdash^{\vec{x}} \varphi$\/ is $\T$-provable, then we simply say that $\T$\/ proves $\varphi$, and write $\T \vdash^{\vec{x}} \varphi$. 

We refer the reader to~\cite[Definition 1]{Palmgren} for the logical axioms and inference rules of PHL. 
The distinguishing feature of this deduction system is that equality of terms is \emph{not} assumed to be reflexive, 
i.e. if $t(\vec{x})$\/ is a term over a given signature, then $\top \vdash^{\vec{x}} t(\vec{x}) = t(\vec{x})$\/ is \emph{not} 
a logical axiom of partial Horn logic, unless $t$\/ is a variable. In other words, if we abbreviate the equation $t = t$\/ by 
$t \downarrow$ (read: $t$ \emph{is defined}), then unless $t$ is a variable, the sequent $\top \vdash^{\vec{x}} t \downarrow$\/ is \emph{not} a logical axiom of PHL. Furthermore, the logical inference rule of term substitution is then only formulated for \emph{defined} terms.

\begin{example}
We have the following examples of quasi-equational theories:
\begin{itemize}
\item Every single-sorted algebraic theory is a quasi-equational theory; this includes
 the usual algebraic theories of (commutative) monoids, (abelian) groups, (commutative) unital rings, etc.
\item The theories of (small) categories, groupoids, categories with a (chosen) terminal object, categories with (chosen) finite products, 
categories with (chosen) finite limits, locally cartesian closed categories, and elementary toposes, 
can all be axiomatized as quasi-equational theories over a two-sorted signature (with one sort $O$ for objects and one sort $A$ for arrows). 
For details 
see \cite[Example 4 and Section 6]{Palmgren}. 
The theory of (small) strict monoidal categories can also be axiomatized as a quasi-equational theory (see Section \ref{sec:picard} below).  
\item If $\T$ is any quasi-equational theory and $\J$ is any small category, then one can axiomatize the functor category $\Tmod^\J$ by a quasi-equational theory $\TJ$, see \cite[Chapter 5]{thesis}. \qed
\end{itemize}
\end{example}  

\noindent In the remainder of the paper, by \emph{theory} we shall mean \emph{quasi-equational theory}, unless explicitly stated otherwise. 

We now review the set-theoretic semantics of PHL. This follows the standard pattern of algebraic theories, 
with the key difference being that function symbols are now only interpreted 
as \emph{partial} functions. We write $f : A \rightharpoondown B$ for a partial function
from $A$\/ to $B$, which is by definition a \emph{total} function $f : \mathsf{dom}(f) \to B$ for some subset $\mathsf{dom}(f) \subseteq A$.
If $\Sigma$\/ is a signature, then a $\Sigma$\emph{-structure} $M$ is a family of sets $M_C$\/ indexed by the sorts $C$\/ of $\Sigma$, together with
interpretations of the function symbols $f:A_1 \times \cdots \times A_k \to A$\/ as partial functions $f^M:M_{A_1} \times \cdots \times M_{A_k} \rightharpoondown M_A$.
By induction on the structure of a term $t$\/ in variable context $x_1:A_1, \ldots, x_k:A_k$, 
we obtain its interpretation as a partial function $t^M: M_{A_1} \times \cdots \times M_{A_k} \rightharpoondown M_A$\/
 in a $\Sigma$-structure $M$, while a Horn formula $\varphi(x_1, \ldots, x_k)$\/ is interpreted as a subset 
 $\varphi(x_1, \ldots, x_k)^M \subseteq M_{A_1} \times \ldots \times M_{A_k}$.

A $\Sigma$-structure $M$\/ \emph{satisfies} a Horn sequent $\varphi \vdash^{\vec{x}} \psi$\/ if $ \varphi(x_1, \ldots, x_k)^M \subseteq \psi(x_1, \ldots, x_k)^M$.
When $\T$\/ is a theory, then a $\Sigma$-structure $M$\/ is a $\T$-\emph{model} when it satisfies all the $\T$-axioms, 
and hence all the $\T$-provable sequents (by soundness of partial Horn logic). 

\begin{definition} {\em Let $\Sigma$\/ be a signature and $M,N$\/ $\Sigma$-structures. 
A \emph{homomorphism} $h:M \to N$\/ is a family of  total functions $h = (h_A : M_A \to N_A)_{A : \Sort}$\/ with the property that 
if $f : A_1 \times \ldots \times A_n \to A$ is any function symbol of $\Sigma$\/ and $(a_1, \ldots, a_n) \in \mathsf{dom}\left(f^M\right)$, 
then $\left(h_{A_1}(a_1), \ldots, h_{A_n}(a_n)\right) \in \mathsf{dom}\left(f^N\right)$ and 
$ h_A\left(f^M(a_1, \ldots, a_n)\right) = f^N(h_{A_1}(a_1), \ldots, h_{A_n}(a_n))$. The homomorphism $h$ \emph{reflects definedness} if 
moreover $\left(h_{A_1}(a_1), \ldots, h_{A_n}(a_n)\right) \in \mathsf{dom}\left(f^N\right)$ always implies $(a_1, \ldots, a_n) \in \mathsf{dom}\left(f^M\right)$. \qed }
\end{definition} 
When working with homomorphisms we often suppress the sort subscripts. The $\T$-models and their homomorphisms 
then form a category $\Tmod$, which is complete and cocomplete.

\begin{definition}
{\em A \emph{morphism of theories} $\rho:\T \to \mS$\/ consists of a mapping $A \mapsto \rho(A)$\/ from the sorts of $\T$\/ to the sorts of $\mS$ 
and a mapping $f \mapsto \rho(f)$\/ from the function symbols of $\T$\/ to the terms of $\mS$\/ that preserves both typing and provability. \qed
}
\end{definition}

When $\rho:\T \to \mS$\/ is a morphism of theories, we have an
induced functor $\rho^*:\Smod \to \Tmod$ by \cite[Proposition 28]{Palmgren}. This functor $\rho^*$ sends an 
$\mS$-model $M$ to the $\T$-model $\rho^*M$ with $\left(\rho^*M\right)_A := M_{\rho(A)}$ for each sort 
$A$ of $\T$ and $f^{\rho^*M} := \rho(f)^M$ for each function symbol $f$ of $\T$. In particular, for every sort 
$A$\/ of $\T$\/ there is a forgetful functor $U_A:\Tmod \to \set$\/ sending a model $M$\/ to the 
carrier set $M_A$ (induced by the theory morphism from the single-sorted empty theory to $\T$ that sends the unique sort of the former theory to the sort $A$).
Each such functor also has a left adjoint $F_A$\/ (see e.g.~\cite[Theorem 29]{Palmgren}), 
giving for a set $X$\/ the free $\T$-model $F_A(X)$\/ generated by $X$: $F_A \dashv U_A : \set \rightleftarrows \Tmod$. 
\begin{definition}
{\em For a $\T$-model $M$, we can form the extension $\T(M)$, the \emph{diagram theory of} $M$, adapted from ordinary model theory. It is the extension of $\T$\/ by 
\begin{itemize}
\item A constant $\overline{a}:A$\/ and an axiom $\top \vdash \overline{a} \downarrow$ for every element $a \in M_A$ (for every sort $A$).
\item An axiom $\top \vdash \overline{f(a_1, \ldots, a_k)}=f(\overline{a_1}, \ldots, \overline{a_k})$\/ for every 
function symbol \\ $f : A_1 \times \cdots \times A_k \to A$\/ and tuple $(a_1, \ldots, a_k) \in \dom\left(f^M\right)$. \qed
\end{itemize}
}
\end{definition}

\noindent For better readability, we will generally omit the bar notation on constants of $M$. Clearly $M$\/ is a model of $\T(M)$, and in fact it is the \emph{initial} model: 
$\T(M)\mathsf{mod} \simeq M/\Tmod$ (see \cite[Lemma 2.2.4]{thesis} for a proof).
The obvious theory morphism $\T \to \T(M)$\/ corresponds to the forgetful functor $M/\Tmod \to \Tmod$. 

One of the central constructions in the present paper is that
of \emph{adjoining an indeterminate} to a model. Given a $\T$-model $M$\/ and a sort $A$\/ of $\T$,
we form a new model $M\la \x_A \ra$\/ which is the result of freely adjoining a new element $\x_A$\/ of sort $A$ to $M$. Formally, one can define
$M\la \x_A \ra$\/ as $M+F_A(1)$, where $F_A(1)$\/ is the free $\T$\/-model on one generator of sort $A$. 
Consequently, homomorphisms $M\la \x_A \ra \to N$\/ are in natural bijective correspondence with
pairs $(h,n)$\/ consisting of a homomorphism $h: M \to N$\/ and an element $n \in N_A$. We will 
write $\T(M,\x_A)$\/ for the theory extending the diagram theory $\T(M)$ by a new constant $\x_A : A$ 
and a new axiom $\top \vdash \x_A \downarrow$. One can then equivalently define the $\T$-model 
$M\la \x_A \ra$ as the initial model of $\T(M, \x_A)$. For a sequence of (not necessarily distinct) sorts $A_1, \ldots, A_k$, 
we will also write $\T(M,\x_1, \ldots, \x_k)$ for the theory extending $\T(M)$ by new, pairwise distinct constants 
$\x_i : A_i$ and axioms $\top \vdash \x_i \downarrow$ for each $1 \leq i \leq k$. 

Finally, we note that for a $\T$-model $M$\/, an indeterminate $\x_A$\/ of sort $A$, and an arbitrary sort $B$, we have 
\begin{equation}\label{eq:terms}
 M\la \x_A \ra_B=\left\{t \in \Term^c\left(\T(M),\x_A\right) \mid t : B \text{ and } \T(M, \x_A) \vdash t \downarrow\right\}/{=}, 
\end{equation}
i.e. the carrier set $M\la \x_A \ra_B$ is the quotient of the set of provably defined closed terms 
of sort $B$, possibly containing $\x_A$ and constants from $M$, modulo the partial congruence relation of 
$\T(M,\x_A)$-provable equality. For more details, see \cite[Remark 2.2.7]{thesis}.  

\section{Isotropy}
We now embark on the syntactic description of the covariant isotropy group of a theory.
First, let us briefly review the simpler situation of a single-sorted equational theory $\T$. That is, we describe the isotropy 
group of a $\T$-model $M$ (details are in~\cite{HPS}).
The elements of the model $M\la \x\ra$\/ (for $\x$\/ an indeterminate) can be described 
explicitly as congruence classes of terms $t(\x)$, built from the indeterminate $\x$, constants from $M$, and the operation symbols of $\T$. 
Two such terms are congruent if they are $\T(M, \x)$-provably equal. For example, if $\T$\/ is the theory of monoids
and $M$\/ is a monoid with $m_1,m_2,m_3 \in M$, unit $e$\/, and $m_1m_2=m_3$, then the terms $t=\x m_1\x m_1m_2\x$\/ and $\x em_1e \x em_3 \x$ are congruent. 

For a set-theoretic $\T$-model $M$, each congruence class 
$[t] \in M\la \x \ra$\/ can be interpreted as a function $t^M: M \to M$, via substitution into the indeterminate $\x$. We thus have a mapping
\[ M\la \x \ra \rightarrow [M,M]  \; ; \qquad [t] \mapsto t^M \]
where $[M,M]$\/ is the set of functions from $M$\/ to itself (well-definedness follows from soundness of the set-theoretic semantics of equational logic). 
Moreover, this mapping is a homomorphism of monoids, 
where the monoid structure on $M\la \x \ra$\/ is given by substitution: $[t] \cdot [s] := [ t[s/\x] ]$, the unit being $[\x]$. 
We then restrict on both sides to the invertible elements, obtaining a group homomorphism $\inv(M\la \x \ra) \rightarrow \perm(M)$
from the group of substitutionally invertible (congruence classes of) terms to the permutation group of the set $M$. 
However, we do not wish to just consider arbitrary permutations of the set $M$, 
but rather \emph{automorphisms} of the $\T$-model $M$; in fact, we want to consider \emph{inner} automorphisms,
i.e. automorphisms that extend naturally along any homomorphism $M \to N$. On the level of terms $[t] \in M\la \x \ra$, this is achieved by the following definition:
$[t]$\/ is said to \emph{commute generically with} a function symbol $f:A^n \to A$\/ ($A$ being the unique sort of $\T$) if  
\[ \T(M, \x_{1}, \ldots, \x_{n})  \vdash \ t[f(\x_{1}, \ldots, \x_{n})/\x] = f(t[\x_1/\x], \ldots, t[\x_n/\x]).\]
We then form the subgroup $\defaut(M)$\/ of $\inv(M\la \x \ra)$\/ on those $[t]$\/ that commute 
generically with all function symbols of the theory. This ensures
that such a $[t]$\/ induces an \emph{automorphism} of the $\T$-model $M$ and not merely a permutation 
of its underlying set, thus yielding a mapping $(-)^M:\defaut(M) \to \aut(M)$. 
However, it turns out that such an automorphism induced by an element of $\defaut(M)$ is also \emph{inner}. Indeed,
given $h:M \to N$, we obtain a homomorphism $h\la \x \ra: M\la \x \ra \to N \la \x \ra$ of the substitution monoids, 
which restricts to a group homomorphism $\defaut(M) \to \defaut(N)$. 
It can then be shown that the subgroup $\defaut(M)$\/ is isomorphic to the covariant isotropy group of $M$, where $\theta_M: \Z(M) \to \aut(M)$\/ is the comparison homomorphism~\eqref{eq:comp}:
\[ 
\xymatrix{
& \defaut(M) \ar[d]^{(-)^M} \ar[r]_\subseteq & \inv(M\la \x \ra) \ar[d]^{(-)^M}
\\
 \Z(M) \ar@{-->}[ur]^\cong \ar[r]_{\theta_M} & \aut(M) \ar[r]_\subseteq  & \perm(M) 
 }
\]
We now explain how to extend this result to a (multi-sorted) \emph{quasi-equational theory} $\T$. 
First of all, in order to accommodate multi-sortedness, we need to consider,
for a $\T$-model $M$, the model $M\la \x_A \ra$ obtained by adjoining an indeterminate $\x_A$\/ of sort $A$ for any sort $A$ of $\T$. From 
the fact  that under the interpretation mapping $t \mapsto t^M$\/ substitution corresponds to composition, it follows 
 that $M\la \x_A \ra_A$\/ carries a monoid structure, defined as before in terms of substitution into the indeterminate 
 $\x_A$. We now write 
\[ M\la \bar{\x} \ra := \prod_{A:\Sort} M\la \x_A \ra_A\] for the sort-indexed product monoid of these substitution monoids.
An element of $M\la \bar{\x} \ra$\/ is therefore a sort-indexed family of congruence classes of terms 
$[s_A]_A$\/, where $s_A \in \Term^c(\T(M), \x_A)$ is of sort $A$ and $\T(M, \x_A) \vdash s_A \downarrow$. 
Given such a tuple $[s_A]_A$, its interpretation gives us, at each sort $A$, a \emph{total} function 
$s_A^M: M_A \to M_A$ (because $s_A$ is provably defined in $\T(M, \x_A)$), defined via substitution into the indeterminate $\x_A$ (cf. \cite[Remark 2.2.12]{thesis}).
The central definitions towards characterizing those $[s_A]_A \in M\la \bar{\x} \ra$\/ that induce elements of isotropy for $M$ are then as follows:

\begin{definition}
\label{commutesgenericallydefn}
{\em Let $M$\/ be a $\T$-model and $[s_C]_C \in  M\la \bar{\x}\ra$. 
\begin{itemize}
\item If $f : A_1 \times \ldots \times A_n \to A$ is a function symbol of $\Sigma$, then we say that $([s_C])_C$ \emph{commutes generically with} 
$f$\/ if the Horn sequent 
\[ f(\mathsf{x}_1, \ldots, \mathsf{x}_n) \downarrow \ \vdash \ s_A[f(\mathsf{x}_{1}, \ldots, \mathsf{x}_{n})/\mathsf{x}_A] = f\left(s_{A_1}[\mathsf{x}_1/\mathsf{x}_{A_1}], \ldots, s_{A_n}[\mathsf{x}_n/\mathsf{x}_{A_n}]\right) \] 
is provable in $\T(M, \mathsf{x}_{1}, \ldots, \mathsf{x}_{n})$. 
\item We say that $([s_C])_C$ is \emph{invertible} if for each sort $A$ there is some $\left[s_A^{-1}\right] \in M\la \mathsf{x}_A \ra_A$ with \[ \T(M, \mathsf{x}_A) \vdash s_A\left[s_A^{-1}/\mathsf{x}_A\right] = \mathsf{x}_A = s_A^{-1}[s_A/\mathsf{x}_A]. \]
\item We say that $([s_C])_C$ \emph{reflects definedness} if for every function symbol $f : A_1 \times \ldots \times A_n \to A$ in $\Sigma$ with $n \geq 1$, the sequent  
\[ f\left(s_{A_1}[\mathsf{x}_1/\mathsf{x}_{A_1}], \ldots, s_{A_n}[\mathsf{x}_n/\mathsf{x}_{A_n}]\right) \downarrow \ \ \vdash f(\mathsf{x}_{1}, \ldots, \mathsf{x}_{n}) \downarrow \] is provable in $\T(M, \mathsf{x}_{1}, \ldots, \mathsf{x}_{n})$. \qed
\end{itemize}
}
\end{definition}

\noindent The condition that $[s_C]_C$\/ commutes generically with the function symbols of $\T$\/ then ensures that $[s_C]_C$\/ induces not just an endofunction of each carrier set $M_C$\/ but in fact an endo\emph{morphism} of the $\T$-model $M$. Invertibility of $[s_C]_C$ then ensures that these endomorphisms are bijective. However, due to the fact that function symbols
are interpreted as partial maps, a (sortwise) bijective homomorphism is not in general an isomorphism in $\Tmod$: a bijective homomorphism is an isomorphism precisely when it
reflects definedness (cf. \cite[Lemma 2.2.33]{thesis}). Thus, the third condition ensures that the inverses $\left[s_A^{-1}\right]$\/
also induce endomorphisms. 

Let us write $\defaut(M)$\/ again for the subgroup of the product monoid $M\la \bar{\x} \ra$\/ consisting of those elements satisfying the three conditions above. 
We then have the following characterization, of which detailed proofs can be found in
~\cite[Theorems 2.2.41, 2.2.53]{thesis}

\begin{theorem}\label{thm:QEiso}
Let $\T$\/ be a quasi-equational theory. Then for any $M \in \Tmod$ we have 
\[ \Z(M) \cong \defaut(M)=\left\{ [s_C]_C \in M\la \bar{\x} \ra \ \rule[-2mm]{.2mm}{6mm} \ [s_C]_C \begin{array}{l} \text{is invertible, commutes generically with} \\ \text{all operations, and reflects definedness.}  
\end{array} \right\}. \qed \]
\end{theorem}

\section{Monoidal categories and the Picard group}\label{sec:picard}
With this description of the isotropy group of an arbitrary quasi-equational theory, we now turn to 
the specific example of strict monoidal categories. We can axiomatize these using the following signature $\Sigma$
(where the first two ingredients comprise the signature for categories):
\begin{itemize}
\item two sorts $O$\/ and $A$\/ (for objects and arrows);
\item function symbols $\dom,\cod:A \to O$,  $\mathsf{id}:O \to A$, and $\circ:A \times A \to A$;
\item function symbols $\otimes_O:O \times O \to O$, $\otimes_A:A \times A \to A$;
\item constant symbols $I_O:O$\/ and $I_A:A$.
\end{itemize}
Whenever reasonable, we omit the subscripts on $\otimes$\/ and $I$. As axioms, we take those 
for categories and add (omitting the hypothesis $\top$):
\begin{itemize}
\item $x \otimes y \downarrow, \qquad I \downarrow$, 
\item $x \otimes (y \otimes z)=(x \otimes y) \otimes z\/, \qquad x \otimes I = x = I \otimes x$,
\item $\dom(f \otimes g)=\dom(f) \otimes \dom(g)\/ , \qquad \cod(f \otimes g)=\cod(f) \otimes \cod(g)$,
\item $(f \otimes g) \circ (h \otimes k)=(f\circ h) \otimes (g \circ k)$,
\item $\mathsf{id}(x \otimes y)=\mathsf{id}(x) \otimes \mathsf{id}(y)\/ , \qquad \mathsf{id}(I_O)=I_A$,
\end{itemize}
where the penultimate axiom of course requires the hypotheses that $f \circ h$ and $g \circ k$ are defined.
Note that in this fragment of logic, we need to include an axiom forcing the tensor products and unit object and arrow to be \emph{total} operations.
Because of strict associativity, we may omit brackets when dealing with nested expressions involving tensor products.
We shall henceforth denote this theory by $\T$, and 
write $\StrMonCat$\/ for its category of models, whose objects are small strict monoidal
categories and whose morphisms are strict monoidal functors. Our goal is now to prove the following:

\begin{theorem}\label{thm:picard}
The covariant isotropy group $\Z(\C)$\/ of a strict monoidal category $\C$\/ is isomorphic to the Picard group of $\C$, i.e. the group of invertible elements in the monoid
of objects of $\C$. \qed 
\end{theorem} 
Because a strict monoidal category is a monoid object in $\cat$, we have two functors
\[ \ob, \arr: \cat(\mon) = \StrMonCat \rightrightarrows \mon. \]
We shall ultimately prove that the diagram
\begin{equation}\label{eq:isocomp}
\xymatrix{
\StrMonCat \ar[rr]^-{\ob} \ar[dr]_\Z && \mon \ar[dl]^{\Z_{\mon}} \\
& \grp
}
\end{equation}
commutes up to isomorphism, showing that the covariant isotropy functor of $\StrMonCat$\/ is 
completely determined by the covariant isotropy functor of $\mon$. Since we have proved in~\cite[Example 4.3]{HPS} 
that the latter sends a monoid $M$\/ to its subgroup of invertible elements, Theorem~\ref{thm:picard} then follows.\footnote{For a general functor $F:\E \to \F$\/ 
it is \emph{not} the case that $\Z_\E \cong \Z_F \circ F$. In fact, in~\cite{FHK} it is explained that in general the relationship between $\Z_\E$\/ and $\Z_\F \circ F$\/ takes
the form of a \emph{span}. The commutativity of~\eqref{eq:isocomp} may thus be expressed by saying that both legs of the span associated with $\ob$\/ are isomorphisms.}

\subsection{Monoidal categories and indeterminates}
In this section we analyse the process of adjoining an indeterminate to a strict monoidal category. Let us first describe
explicitly the result of adjoining an indeterminate to a monoid. 

\begin{definition}{\em
Let $M$\/ be a monoid, and $X$\/ a set of symbols disjoint from $M$. 
\begin{itemize}
\item A \emph{word} over $M\la X \ra$\/ is formal string of symbols from the alphabet $M \cup X$.
\item A word $w$\/ is in \emph{(expanded) normal form}
when it has the form $w=m_0x_0m_1x_1 \cdots x_{n-1}m_n$\/
for $m_i \in M$\/ and $x_j \in X$. In other words, $w$ is in expanded normal form if it contains no two consecutive elements of $M$, and if every occurrence of some $x \in X$ in $w$ is flanked on both sides by an element of $M$. \qed
\end{itemize}}
\end{definition}

\noindent We then have (by taking an arbitrary word, multiplying adjacent elements from $M$\/ and inserting the unit of $M$ whenever necessary):

\begin{lemma}\label{lem:monadj}
When $M=(M,e,\cdot)$\/ is a monoid, every element $w$ of the monoid $M\la \x \ra$\/ has a canonical representative
$w=m_0\x m_1\x \cdots \x m_n$\/ in expanded normal form.

Moreover, the unit of $M\la \x \ra$\/ is represented as the word $e$\/ and multiplication
is given by $(m_0\x m_1\x \cdots \x m_j) \cdot (m'_0\x m'_1\x \cdots \x m'_k)=m_0\x m_1\x \cdots \x (m_j\cdot m'_0) \x m'_1 \cdots \x m'_k$. \qed
\end{lemma}

We now turn turn to the process
of adjoining an indeterminate \emph{object} $\x_O$, i.e. an indeterminate of sort $O$, to a strict monoidal category $\C$. 
In order to determine the objects of $\C\la \x_O \ra$, we note that the functor $\ob: \StrMonCat \rightarrow \mon$\/ has both adjoints:
\[
\xymatrix{
\StrMonCat \ar[rrr]|-{\ob}^\perp_\perp &&& \mon \ar@/_2ex/[lll]_{\Delta} \ar@/^2ex/[lll]^\nabla \\
}
\]
Here $\Delta$\/ sends a monoid $M$\/ to the discrete strict monoidal category on $M$\/ 
and $\nabla$\/ sends $M$\/ to the indiscrete strict monoidal category on $M$. In fact, if $\E$ is \emph{any} category with finite limits, then the forgetful functor $\ob : \cat(\E) \to \E$ has both adjoints (for a proof, mimic the argument for $\E = \set$). 
As such, $\ob : \StrMonCat \to \mon$\/ preserves all limits and colimits. Now by definition $\C\la \x_O\ra \cong \C+F\mathbf{1}$, 
where $F\mathbf{1}$\/ is the free strict monoidal category on a single object;
moreover, the latter is easily seen to be isomorphic to $\Delta(F1)$, the discrete strict monoidal category on the free monoid $F1$ on one generator.
We thus have
\[ \ob(\C\la \x_O\ra) \cong \ob(\C + F\mathbf{1}) \cong \ob(\C)+\ob(F\mathbf{1})=\ob(\C)+F1\cong \ob(\C)\la \x \ra.\]
This shows that the object forgetful functor preserves the process of adjoining an indeterminate of sort $O$.\footnote{Note that for a functor $\rho^*:\mS\mathsf{mod} \to \Tmod$\/ induced by
a theory morphism $\rho:\T \to \mS$\/ it is
not in general the case that $\rho^*(M\la \x \ra) \cong (\rho^*M)\la \x \ra$.}

We now describe the monoid of arrows of $\C\la \x_O \ra$. It is not true that $\arr:\StrMonCat \to \mon$\/ preserves binary coproducts. However, it \emph{does} preserve
the coproduct $\C+F\mathbf{1}$.

\begin{lemma}
If $\C \in \StrMonCat$, we have $\arr(\C\la \x_O \ra) \cong \arr(\C)\la \x \ra$.
\end{lemma}

\begin{proof}
We sketch a syntactic proof, noting that the result can also be deduced categorically from the fact
that the endofunctor $-+F1: \mon \to \mon$\/ preserves pullbacks.

An element of $\arr(\C \la \x_O \ra)$\/ is a congruence class of terms $t$\/ built up from the operations of $\T$, arrows of $\C$\/, 
and the term $\mathsf{id}(\x_O)$. One shows by induction that every such term $t$\/ is congruent to one of the form
$t = f_1 \otimes \mathsf{id}(\x_O) \otimes f_2 \otimes \mathsf{id}(\x_O) \otimes \cdots \otimes \mathsf{id}(\x_O) \otimes f_n$\/ where
each $f_i$\/ is an arrow of $\C$. Thus, the monoid $\arr(\C \la \x_O \ra)$\/ is isomorphic, by Lemma~\ref{lem:monadj},
to $\arr(\C)\la \x \ra$.
\end{proof} 

\noindent In fact, we may describe the relationship between the functor $(-)+F\mathbf{1}$\/ 
adjoining an indeterminate object to a strict monoidal category and the functor $(-)+F1$\/ adjoining an indeterminate element to a monoid
as follows.

\begin{proposition}
The functor $(-)+F\mathbf{1}: \cat(\mon) \to \cat(\mon)$\/ is naturally isomorphic to $\cat(- + F1)$. \qed
\end{proposition}

\noindent We thus obtain the following explicit description of the strict monoidal category $\C\la \x_O\ra$: 

\begin{description}
\item {\bf Objects: } Words $a_1\x a_2 \x \cdots \x a_n$\/ where each $a_i$\/ is an object of $\C$.
\item {\bf Morphisms: } Words $f_1\x f_2\x \cdots \x f_n$\/ where each $f_i$\/ is an arrow of $\C$.
\item {\bf Domain: } $\dom(f_1\x  \cdots \x f_n)=\dom(f_1)\x \cdots \x \dom(f_n)$.
\item {\bf Codomain: } $\cod(f_1\x  \cdots \x f_n)=\cod(f_1)\x \cdots \x \cod(f_n)$.
\item {\bf Identities: } $\mathsf{id}(a_1\x \cdots \x a_n)=\mathsf{id}(a_1)\x \cdots \x \mathsf{id}(a_n)$.
\item {\bf Composition: } $(f_1\x \cdots \x f_n) \circ (g_1 \x \cdots \x g_n)= f_1g_1 \x \cdots \x f_ng_n$.
\item {\bf Tensors: } $(a_1\x \cdots \x a_n) \otimes (b_1 \x \cdots \x b_m)=a_1 \x \cdots \x (a_n \otimes b_1) \x \cdots \x b_m$.
\item {\bf Tensor units: } $I_O, I_A$\/ (tensor units of $\C$\/ regarded as one-letter words). \qed
\end{description} 

Next, we address the issue of adjoining an indeterminate \emph{arrow} $\x_A$\/ to $\C$. Here we cannot invoke a simple categorical fact
about coproducts, because $\arr:\StrMonCat \to \mon$\/ does not preserve coproducts of the relevant kind (which, to be explicit, is
coproducts with the free strict monoidal category $F\mathbf{2}$, where $\mathbf{2}$\/ is the free-living arrow). We are thus forced to carry out a direct 
\emph{syntactic} analysis of the objects and arrows of $\C\la \x_A \ra$. Note that these are generated, under the operations of domain, codomain, 
identities, composition, and tensor product, from the objects and arrows of $\C$, together with the new arrow $\x_A$. In particular,
there will be two new objects $\dom(\x_A)$\/ and $\cod(\x_A)$, and corresponding identity arrows $\mathsf{id}(\dom(\x_A))$, $\mathsf{id}(\cod(\x_A))$. 

\begin{definition}{\em
Let $\C \in \StrMonCat$. A closed term $t \in \Term^c(\C, \x_A)$\/ of sort $O$\/ is in \emph{normal form} when 
it is of the form $t=a_1 \otimes \x_1 \otimes  \cdots  \otimes \x_{k-1} \otimes a_k$\/, where
each $a_i$\/ is an object of $\C$ and each $\x_i  \in \{\dom(\x_A), \cod(\x_A)\}$. 
A closed term $t \in \Term^c(\C, \x_A)$\/ of sort $A$\/ is in \emph{normal form} when 
it is of the form $t=f_1 \otimes \x_1 \otimes \cdots \otimes \x_{k-1} \otimes f_k$\/, where each 
$f_i$\/ is an arrow of $\C$\/ and each $\x_i \in \{\x_A, \mathsf{id}(\dom(\x_A)), \mathsf{id}(\cod(\x_A))\}$. \qed 
}
\end{definition}

We may now describe $\C\la \x_A\ra$\/ in terms of normal forms. It is straightforward to prove, by directly verifying the universal property,
 that the category described below is indeed isomorphic to $\C \la \x_A \ra$. Alternatively, one can endow the set 
 $\{t \in \Term^c(\C,\x_A) \mid t \downarrow\}$\/ with a rewriting
 system and show that each term has a unique normal form.

\begin{description}
\item {\bf Objects: } closed terms of sort $O$\/ in normal form.
\item {\bf Arrows: } closed terms of sort $A$\/ in normal form.
\item {\bf Domain: } $\dom(f_1 \otimes \x_1 \otimes  \cdots  \otimes \x_{k-1} \otimes f_k)=\dom(f_1) \otimes \y_1 \otimes\cdots \otimes \y_{k-1} \otimes \dom(f_k)$\/ 
where $\y_i=\dom(\x_A)$\/ when $\x_i=\x_A$\/ or $\x_i=\mathsf{id}(\dom(\x_A))$, and $\y_i=\cod(\x_A)$\/ otherwise.
\item {\bf Codomain: } $\cod(f_1 \otimes \x_1 \otimes  \cdots  \otimes \x_{k-1} \otimes f_k)=\cod(f_1) \otimes \y_1 \otimes\cdots \otimes \y_{k-1} \otimes \cod(f_k)$\/ 
where $\y_i=\cod(\x_A)$\/ when $\x_i=\x_A$\/ or $\x_i=\mathsf{id}(\cod(\x_A))$, and $\y_i=\dom(\x_A)$\/ otherwise.
\item {\bf Identities: } $\mathsf{id}(a_1 \otimes \x_1 \otimes  \cdots  \otimes \x_{k-1} \otimes a_k)=\mathsf{id}(a_1) \otimes \mathsf{id}(\x_1) \otimes \cdots \otimes \mathsf{id}(\x_{k-1})\otimes \mathsf{id}(a_k)$.
\item {\bf Composition: } For $t=f_1 \otimes \x_1 \otimes \cdots \otimes \x_{k-1} \otimes f_k$ and $s = g_1 \otimes \x_1' \otimes \cdots \otimes \x_{k-1}' \otimes g_k$\/ with
$\cod(t)=\dom(s)$,
define $s \circ t = (g_1f_1) \otimes \z_1 \otimes \cdots \otimes \cdots \otimes \z_{k-1} \otimes (g_kf_k)$, where $\z_i$ is defined from $\x_i$ and $\x_i'$ in the evident way.
\item {\bf Tensors: } $(a_1\otimes \x_1 \otimes  \cdots \otimes \x_{n-1}\otimes a_n) \otimes (b_1 \otimes \y_1 \otimes \cdots  \otimes \y_{m-1} \otimes b_m)
= \\ a_1\otimes  \x_1 \otimes \cdots\otimes  \x_{n-1} \otimes (a_n \otimes b_1) \otimes \y_1 \otimes \cdots \otimes \y_{m-1} \otimes b_m$.
\item {\bf Tensor units: } $I_O, I_A$\/ (tensor units of $\C$\/ regarded as one-letter words).
\end{description}

\subsection{Isotropy group}
We are now in a position to analyse the isotropy group of a strict monoidal category. By the results of the previous section, we know that an element of isotropy
of a strict monoidal category $\C$\/ may be taken to be of the form $(s_O,s_A)$, where $s_O$\/ and $s_A$\/ are closed terms in normal form of sort $O$\/ and $A$\/ 
respectively. 

The first observation is that elements of isotropy of the monoid $\ob(\C)$\/ induce elements of isotropy of $\C$. (As we shall see in the next section, this is not specific to
strict monoidal categories.) In what follows, we write $\Z(\C)$\/ for the isotropy group of a strict monoidal category $\C$\/, and $\Z_\mon(M)$\/ for the isotropy group of
a monoid $M$\/ (which is the group of invertible elements of $M$ by \cite[Example 4.3]{HPS}).

\begin{lemma}\label{lem:obtocat}
Let $\C \in \StrMonCat$. When $a$\/ is an invertible object in the monoid $\ob(\C)$\/ with inverse $b$, the pair 
$(a \otimes \x_O  \otimes b, \mathsf{id}(a) \otimes \x_A \otimes \mathsf{id}(b))$\/ is an element of $\Z(\C)$.
\end{lemma}
\begin{proof}
To show that $(a \otimes \x_O  \otimes b, \mathsf{id}(a) \otimes \x_A \otimes \mathsf{id}(b))$\/ is an element of isotropy, one can straightforwardly 
verify that it is invertible, commutes generically with all operations
of $\T$, and reflects definedness (for details, see \cite[Proposition 3.9.35]{thesis}). However, it is less work to show 
directly that given a strict monoidal functor
$F:\C \to \D$\/, we obtain an automorphism $\alpha_F$\/ of $\D$\/ as follows. 
On objects we set $\alpha_F(d)=Fa \otimes d \otimes Fb$, while on arrows we set $\alpha_F(f)=\mathsf{id}(Fa) \otimes f \otimes \mathsf{id}(Fb)$.
It is routine to check that this defines an automorphism and that the family $\alpha_F$\/ is natural in $F$.
\end{proof} 

The above lemma gives us a mapping $\theta_\C:\Z_{\mon}(\ob(\C)) \to \Z(\C)$. 
It is easily verified that this is in fact a group homomorphism, natural in $\C$. 

Next, we define a retraction $\sigma$\/ of $\theta$. 
This is done categorically using the right adjoint $\nabla$\/ to $\ob$. Concretely, given an element of isotropy
$\alpha \in \Z(\C)$, we define an element $\sigma_\C(\alpha) \in \Z_\mon(\ob(\C))$\/ as follows:
consider a monoid homomorphism $h:\ob(\C) \to N$. 
This corresponds by the adjunction $\ob \dashv \nabla$ to a strict monoidal functor $\tilde{h}: \C \to \nabla(N)$; the component of $\alpha$\/ at $\tilde{h}$\/ is an
automorphism of $\nabla(N)$, whence $\ob\left(\alpha_{\tilde{h}}\right)$\/ is an automorphism of $N$\/ (using the fact that $\ob \circ \nabla=1$). This leads to:

\begin{lemma}
If $\C \in \StrMonCat$, the map $\sigma_\C: \Z(\C) \to \Z_\mon(\ob(\C))$ is a group homomorphism. \qed
\end{lemma}

Interpreting this syntactically, we find that if $(s_O, s_A) \in \Z(\C)$, then 
$s_O \in \Z_\mon(\ob(\C))$, and hence $s_O=a \otimes \x_O \otimes b$\/ for an invertible 
object $a$ with inverse $b$. We also see that $\sigma_\C$\/ is a retraction of $\theta_\C$, i.e. that $\sigma_\C \circ \theta_\C = 1$. 

Since $\theta_\C$ is a section, it now remains to show that $\theta_\C$ is an epimorphism of groups, i.e. is surjective. So we must show for any element of isotropy
$(s_O,s_A)=(a \otimes \x_O \otimes b,s_A) \in \Z(\C)$ (with invertible object $a$ and inverse $b$) that we have $s_A=\mathsf{id}(a) \otimes \x_A \otimes \mathsf{id}(b)$.  
To this end, we first note that since $(s_O,s_A)$\/ commutes generically
with the operations $\dom$ and $\cod$\/ we get
\[ a \otimes \dom(\x_A) \otimes b = s_O[\dom(\x_A)/\x_O] = \dom(s_A)\]
and likewise
\[  a \otimes \cod(\x_A) \otimes b = s_O[\cod(\x_A)/\x_O] = \cod(s_A).\]
Thus, by uniqueness of normal forms, $s_A$\/ must have the form $f \otimes \x_A \otimes g$\/ for some morphisms $f:a \to a$\/ and $g:b \to b$\/ of $\C$.
So we must now show that $f=\mathsf{id}(a)$\/ and $g=\mathsf{id}(b)$, 
and for that we use the fact that $(s_O,s_A)$\/ commutes generically with $\mathsf{id}$, giving
\[ f \tensor \mathsf{id}(\x_O) \tensor g = s_A[\mathsf{id}(\x_O)/\x_A] = \mathsf{id}(s_O) = \mathsf{id}(a \otimes \x_O \otimes b) = \mathsf{id}(a) \otimes \mathsf{id}(\x_O) \otimes \mathsf{id}(b).\]
We now get the desired equalities $f=\mathsf{id}(a)$\/ and $g=\mathsf{id}(b)$\/ by appealing to the uniqueness of normal forms. 
This concludes the proof of Theorem~\ref{thm:picard}.

\section{Further examples and applications}
In this section we briefly explore some further theories of interest,
and indicate the extent to which the analysis of the case of strict monoidal categories can be generalized.

\subsection{Internal categories}
The analysis of strict monoidal categories reveals that it is profitable, at least for the purposes of understanding isotropy,
to regard strict monoidal categories as internal categories in the category $\mon$ of monoids. This naturally raises the following question:
are there other algebraic theories $\T$\/ for which the forgetful functor $\ob: \cat(\Tmod) \to \Tmod$\/ induces an isomorphism on the level of
isotropy groups? 

Let us first state which of the ideas from the case of monoids carry over to a general algebraic theory $\T$.
First of all, we still have a string of adjunctions
\[
\xymatrix{
\cat(\Tmod) \ar[rrr]|-{\ob}^\perp_\perp &&& \Tmod \ar@/_2ex/[lll]_{\Delta} \ar@/^2ex/[lll]^\nabla \\
}
\]
with $\ob \circ \nabla  \cong 1 \cong \ob \circ \Delta$. This allows us to deduce the existence of
a pair of natural comparison homomorphisms
\[ \theta_\C: \Z_\T(\ob(\C)) \to \Z(\C) \; ; \qquad \sigma_\C:\Z(\C) \to \Z_\T(\ob(\C))\] with 
$\sigma \circ \theta = 1$ (here $\Z$\/ denotes the isotropy of $\cat(\Tmod)$ and $\Z_\T$\/ that of $\Tmod$).
We thus have:

\begin{lemma}
Let $\T$\/ be any algebraic theory and $\C$\/ any internal category in $\Tmod$. Then $\Z_\T(\ob(\C))$\/ is a retract of $\Z(\C)$, naturally in $\C$.
\end{lemma}

\noindent In the case of strict monoidal categories, we were able to prove syntactically that the embedding-retraction pair $(\theta,\sigma)$\/ 
is an isomorphism. The same proof can also be applied in two other cases of interest:

\begin{proposition}
The isotropy group of a crossed module $A \to G$\/ is isomorphic to $G$. \qed 
\end{proposition}

\begin{proof}
When composing the functor $\ob:\cat(\grp) \to \grp$\/ with the equivalence $\xmod \xrightarrow{\sim} \cat(\grp)$, one
obtains the forgetful functor which sends a crossed module $A \to G$\/ to $G$. Moreover, the isotropy group of 
a group $G$\/ is $G$\/ itself by \cite[Example 4.1]{HPS}.
\end{proof}

\begin{proposition}
The isotropy group of a strict \emph{symmetric} monoidal category is trivial.
\end{proposition}

\begin{proof}
The isotropy group of commutative monoids is trivial by \cite[Example 4.4]{HPS}.
\end{proof}

\subsection{Presheaf categories}\label{subsec:psh}
Using Theorem \ref{thm:QEiso}, we can also compute the covariant 
isotropy of any presheaf category $\set^\J$\/ for a small category $\J$. 
We first axiomatize $\set^\J$\/ as a quasi-equational theory.

\begin{definition}[\textbf{Presheaf Theory}]
{\em Let $\J$ be a small category. We define the signature $\Sigma_\J$\/  to have one sort 
$X_i$ for each $i \in \ob(\J)$ and one function symbol $\alpha_f : X_i \to X_j$ for each arrow $f : i \to j$ in $\J$. 

We define the \emph{presheaf theory} $\TJ$ to be the quasi-equational theory over the signature $\Sigma_\J$ with the following axioms:
\begin{itemize}
\item $\top \vdash^{x : X_i} \alpha_f(x) \downarrow$ for any $f : i \to j$ in $\J$ (i.e. each $\alpha_f$ is \emph{total}). 
\item $\top \vdash^{x : X_i} \alpha_{\mathsf{id}_i}(x) = x$ for every $i \in \ob\J$ (i.e. each $\alpha_{\mathsf{id}_i}$ acts as an identity). 
\item $\top \vdash^{x : X_i} \alpha_g(\alpha_f(x)) = \alpha_{g \circ f}(x)$ for any composable pair $i \xrightarrow{f} j \xrightarrow{g} k$ in $\J$. \qed 
\end{itemize}
}
\end{definition}

\noindent We will lighten notation and write $i$\/ instead of $X_i$ and $f$\/ instead of $\alpha_f$. 
We write $\x_i$\/ for an indeterminate of sort $i$. It is completely straightforward to verify 
that we have an isomorphism of categories $\TJmod \cong \set^\J$ (for details, see \cite[Proposition 5.1.8]{thesis}). 
So to compute the covariant isotropy group $\Z_{\set^\J} : \set^\J \to \grp$\/ of the category $\set^\J$, it
is equivalent to compute the covariant isotropy group $\Z_{\TJ} : \TJmod \to \grp$\/ of the theory $\TJ$.

According to Theorem~\ref{thm:QEiso}, we have
for a $\TJ$-model (i.e. functor) $F:\J \to \set$ that
\[ \Z(F)\cong \left\{[s_i] _{i} \in \prod_{i \in \J} F\la \x_i \ra_i \mid [s_i]_i \text{ is invertible and commutes gen. with all $f:i \to j$\/} \right\}.\]
Note that since all terms are provably total in $\TJ$, the condition that $[s_i]_i$\/ reflects definedness can be omitted. 
We now require the following preparatory lemma.

\begin{lemma}
\label{arrowequality}
Let $M \in \TJmod$. If $f, f' : i \to j$ are parallel arrows in $\J$ and $\TJ(M, \x_i) \vdash f(\x_i) = f'(\x_i)$, then $f = f'$. 
\end{lemma}

\begin{proof}
Note that the assumption $\TJ(M, \x_i) \vdash f(\x_i) = f'(\x_i)$\/ implies that for any homomorphism (i.e. natural transformation) $\eta: M \to N$\/ we have
$N(f)=N(f')$, since given any $a \in N_i$\/ there is a homomorphism $[\eta,a]:M\la \x_i \ra \to N$\/ sending $\x_i$\/ to $a$ (cf. also \cite[Lemma 3.1.2]{thesis}). 
We take $N:\J \to \set$\/ to be $N:=M+\J(i,-)$\/ and $\eta$\/ to be the coproduct inclusion. Then
$f=f \circ \mathsf{id}(i) = N(f)(\mathsf{id}(i))=N(f')(\mathsf{id}(i))=f' \circ \mathsf{id}(i)=f'$\/ as required.
\end{proof}

\noindent As a consequence of this lemma, we find that any term congruence class $[t] \in \TJ(M,\x_i)$\/ has a unique representation as $t\equiv a$\/ for some $a \in M_j$\/ 
or $t\equiv f(\x_i)$\/ for some $f$ with domain $i$, depending on whether the indeterminate $\x_i$\/ occurs in $t$. 

Let $\Aut(\id_\J)$ be the group of natural automorphisms of the identity functor $\id_\J : \J \to \J$ of a small category $\J$. 
This group is sometimes called the \emph{center} of $\J$.
We now have: 
\begin{proposition}
\label{presheafisotropyprop}
Let $\J$ be a small category. For any $M \in \TJmod$ we have
\[ \Z(M) = \left\{ \left([\psi_i(\x_i)]\right)_i \in \prod_{i \in \J} M\la \x_i \ra_i \colon \psi \in \Aut(\id_\J)\right\}. \]
\end{proposition}

\begin{proof}
It is straightforward to prove the right-to-left inclusion using the assumption that 
$\psi$ is a natural automorphism of $\id_\J$, so let us turn to the less obvious converse inclusion.
So suppose that $([s_i])_{i \in \J} \in \Z_{\TJ}(M) \subseteq \prod_i M \la \x_i \ra_i$. By the lemma, as well as the fact that
invertible terms must contain the indeterminate, we may represent $s_i=\psi_i(\x_i)$, where $\psi_i:i \to i$\/ is a map in $\J$. 
We show that $\psi := (\psi_i)_{i \in \J}$ is a natural automorphism of $\id_\J$. First, each $\psi_i : i \to i$ is an isomorphism: take the inverse $([t_i])_i$\/ of 
$([s_i])_i$, and represent this inverse as $\chi_i(\x_i)$\/ for $\chi_i:i \to i$. Since 
$\TJ(M, \x_i)$ proves the equations $(\psi_i \circ \chi_i)(\x_i) = \psi_i(\chi_i(\x_i)) = \x_i = \mathsf{id}_i(\x_i)$\/
 and $(\chi_i \circ \psi_i)(\x_i) = \mathsf{id}_i(\x_i)$, it follows by Lemma~\ref{arrowequality} that $\psi_i$\/ is the inverse of $\chi_i$.

To show that $\psi$ is natural, let $f : j \to k$ be any arrow in $\J$, and let us show that 
$\psi_k \circ f = f \circ \psi_j$. We know that $\left(\left[\psi_i(\x_i)\right]\right)_i = [s_i]_i$\/ 
commutes generically with the function symbol $f : X_j \to X_k$ of $\Sigma_\J$, which implies that 
$\TJ(M, \x_j) \vdash (\psi_k \circ f)(\x_j) = (f \circ \psi_j)(\x_j)$, from which we 
obtain the required $\psi_k \circ f = f \circ \psi_j$ again by Lemma~\ref{arrowequality}. 
Thus $\psi : \id_\J \xrightarrow{\sim} \id_\J$ is indeed a natural automorphism with $([s_i])_i = \left(\left[\psi_i(\x_i)\right]\right)_i$.                
\end{proof} 

\begin{corollary}
\label{presheafisotropycor}
Let $\J$ be a small category. For any functor $F : \J \to \set$ we have $\Z(F) \cong \Aut\left(\id_\J\right)$, 
and hence the covariant isotropy group functor of $\set^\J$\/ is constant on the automorphism group of $\id_\J$.  
\end{corollary} 

\begin{proof}
Given $([s_i])_{i \in \J} \in \Z_{\TJ}(F)$, we know by Proposition \ref{presheafisotropyprop} that there 
is some $\psi \in \Aut(\id_\J)$ with $[s_i]_i = [\psi_i(\x_i)]_i$. We now show that this assignment 
$([s_i])_i \mapsto \psi$ is a well-defined group isomorphism $\Z_{\TJ}(F) \xrightarrow{\sim} \Aut(\id_\J)$. 
It is well-defined, because if there is also some $\chi \in \Aut(\id_\J)$ with $[s_i]_i = [\psi_i(\x_i)]_i = [\chi_i(\x_i)]_i$, 
then from Lemma \ref{arrowequality} we obtain $\psi = \chi$. It is clearly injective, it is surjective by 
Proposition \ref{presheafisotropyprop}, and it is readily seen to preserve group multiplication, so that it is indeed a group isomorphism. 
\end{proof}

\noindent We can now use Corollary \ref{presheafisotropycor} to characterize the 
covariant isotropy groups of certain presheaf categories of interest.

\begin{proposition}
\label{Gsetisotropy}
If $M$ is a monoid, then the covariant isotropy group $\Z : \set^M \to \grp$ of the category of 
$M$-sets and $M$-equivariant maps is constant on $\mathsf{Inv}(Z(M))$, the subgroup of invertible 
elements of the centre of $M$. In particular, if $G$ is a group, then the covariant isotropy group $\Z : \set^G \to \grp$\/ is constant on $Z(G)$. 
\end{proposition}

\begin{proof}
The result follows immediately from Corollary \ref{presheafisotropycor} and the observation that the 
automorphism group of the identity functor on the monoid $M$, regarded as a one-object category, is isomorphic to $\mathsf{Inv}(Z(M))$.
\end{proof}  

\begin{proposition}
\label{posetisotropy}
Let $\J$ be a \emph{rigid} category, i.e. a category whose objects have no non-identity automorphisms 
(e.g. $\J$ could be a preorder or poset). Then the covariant isotropy group $\Z : \set^\J \to \grp$\/ is trivial. \qed  
\end{proposition}

We point out that Corollary~\ref{presheafisotropycor} illustrates an important difference between covariant isotropy $\set^\J \to \grp$\/ and
contravariant isotropy $\left(\set^\J\right)^{\mathrm{op}} \to \grp$. Indeed, the latter is generally \emph{not} constant, but is a representable functor
$F \mapsto \set^\J[F,Z]$\/ for a suitable presheaf of groups $Z$, that is, an internal group object in $\set^\J$. 
The connection between covariant and contravariant isotropy is then as follows: the group of global sections of $Z$\/ is isomorphic to
the group $\aut(\id_\J)$:
\[ \Gamma(Z) = \set^\J(1, Z) \cong \Z(F) \text{ for } F:\J \to \set .\]

\section{Conclusions and future work}
We have shown how a syntactic description of polymorphic automorphisms
can be fruitfully applied to characterize the covariant isotropy 
of several kinds of structures of relevance in logic, algebra, and computer science. Most notably,
we have shown that the covariant isotropy group of a strict monoidal category coincides with its Picard group
of invertible objects. We have also shown that the covariant isotropy group of a presheaf category $\set^\J$\/
behaves quite differently from the contravariant one, in that it is the constant group with value $\aut(\id_\J)$.

There are several open questions and possible lines for further inquiry: 

\begin{enumerate}
\item The generalization from algebraic to quasi-equational theories presented in this paper 
is the first step on a path upwards through the various fragments of logic. In particular,
we hope to generalize some of the techniques to determine
the isotropy groups of some geometric theories of interest.
\item We have shown how to determine the covariant isotropy groups of presheaf categories,
but we have left open the question of how to determine the isotropy of \emph{sheaf} toposes. 
In particular, it would be of interest to determine the covariant isotropy of the topos of nominal sets
(also known as the Schanuel topos). 
\item For a theory $\T$\/ and diagram category $\J$, there is a
theory $\mS=\mS(\T,\J)$\/ with $\mS\mathsf{mod}\cong\Tmod^\J$ (in Section~\ref{subsec:psh} we considered
the special case where $\T$\/ is the trivial theory, i.e. the theory of sets). 
In~\cite[Chapter 5]{thesis} the second author has
obtained, under mild assumptions on $\T$, a description of the covariant isotropy group of $\TJmod$\/ 
in terms of $\Aut(\id_\J)$ and the isotropy group of $\T$. 
\item We have not yet investigated in detail how isotropy behaves with respect to morphisms of theories $\rho:\T \to \mS$.
(We have seen a rather special case in Section~\ref{sec:picard} with $\ob:\StrMonCat \to \mon$, but the general case
is more involved.) 
\item One possible perspective on the theory of strict monoidal categories is that it is a \emph{tensor product} of the theory of categories
with that of monoids. This leads to the question whether, under suitable conditions on the theories $\T$ and $\mS$, we can 
describe the isotropy of $\T \otimes \mS$\/ in terms of that of $\T$\/ and $\mS$.
\item One can define, for a 2-category $\E$ and object $X \in \E$, the group of natural auto-equivalences of $X/\E \to \E$. This leads to a 2-dimensional version
of isotropy, taking values in 2-groups. It is then possible to show that the 2-isotropy group of a (non-strict) monoidal category
(regarded as an object in the 2-category of monoidal categories and strong monoidal functors) is the Picard 2-group. This will be presented in forthcoming work.  
\end{enumerate}

\bibliography{HPS}

\end{document}